\documentclass[a4paper,USenglish,cleveref, autoref, thm-restate]{lipics-v2021}

\hideLIPIcs  

\title{Communication Requirements for Linearizable
Registers} 

\author{Ra\"issa Nataf}{Technion, Israel} {raissa.nataf@cs.technion.ac.il}{https://orcid.org/0009-0003-1127-754X}{}

\author{Yoram Moses}{Technion, Israel} {moses@ee.technion.ac.il}{https://orcid.org/0000-0001-5549-1781}{}

\authorrunning{R.Nataf, Y.Moses}
\Copyright{Ra\"issa Nataf and Yoram Moses}



\category{} 

\relatedversion{} 

\funding{Yoram Moses is the Israel Pollak academic chair at the Technion.  This work was supported in part by the Israel Science Foundation under grant ISF 2061/19.}

\acknowledgements{We thank Gal Assa, Naama Ben David, and an anonymous referee  for very useful comments that improved the presentation of this paper. We alone are responsible for any errors or misrepresentations.
This is a slightly modified version of a paper \cite{commRequirements2024} with a similar title that appeared in DISC 2024.}

\nolinenumbers 

\usepackage{algorithm2e}
\usepackage{algpseudocode}
\usepackage{amsmath}
\usepackage{amssymb}
\usepackage{amsbsy}
\usepackage{amsthm}
\usepackage{eucal}
\usepackage{bm}
\usepackage{float}
\usepackage{apxproof}

\usepackage{hyperref} 

\algnewcommand\algorithmicforeach{\textbf{for each}}
\algdef{S}[FOR]{ForEach}[1]{\algorithmicforeach\ #1\ \algorithmicdo}
\algdef{SE}[EVENT]{Event}{EndEvent}[1]{\textbf{upon event}\ #1\ \algorithmicdo}{\algorithmicend\ \textbf{event}}%
\algtext*{EndEvent}
\usepackage{calc} 

\usepackage{xcolor}
\usepackage{thmtools}
\newcounter{casenum}

\newcommand{\mc}{\rightsquigarrow}
\newcommand{\opmc}{\boldsymbol{\mc}}

\newcommand{\node}[1]{\langle#1\rangle}

\usepackage[all]{nowidow} 
\newcommand{\shift}[1]{\mathtt{shift}_\Delta[#1]}

\newcommand{\R}{\mathtt{R}}
\newcommand{\W}{\mathtt{W}}
\newcommand{\X}{\mathtt{X}}
\newcommand{\Op}{\mathtt{O}}
\newcommand{\Xa}{\mathtt{X} a}
\newcommand{\Y}{\mathtt{Y}}
\newcommand{\Yb}{\mathtt{Y} b}
\newcommand{\Z}{\mathtt{Z}}
\newcommand{\defemph}[1]{\textbf{\textit{#1}}}
\newcommand{\modelf}{\gamma^f}
\newcommand{\larp}{{\sf l.a.r.p.\/}}

\newcommand{\act}{\alpha}
\newcommand{\RP}{R_Q}
\newcommand{\loc}{\approx}
\newtheorem{notation}{Notation}
\newcommand{\ret}{\mathtt{return}}
\newcommand{\inv}{\mathtt{invoke}}
\newcommand{\move}{\mathtt{move}}
\newcommand{\skipp}{\mathtt{skip}}
\newcommand{\del}{\mathtt{deliver}}
\newcommand{\past}{\mathsf{past}}
\newcommand{\chan}[1]{\mathsf{chan}_{#1}}
\EventEditors{Dan Alistarh}
\EventNoEds{1}
\EventLongTitle{38th International Symposium on Distributed Computing (DISC 2024)}
\EventShortTitle{DISC 2024}
\EventAcronym{DISC}
\EventYear{2024}
\EventDate{October 28--November 1, 2024}
\EventLocation{Madrid, Spain}
\EventLogo{}
\SeriesVolume{319}
\ArticleNo{20}
\begin{document}

\maketitle

\begin{abstract}
While linearizability is a fundamental correctness condition for distributed systems, ensuring the linearizability of implementations can be quite complex. 
An essential aspect of linearizable implementations of concurrent objects is the need to preserve the real-time order of operations. In many settings, however, processes cannot determine the precise timing and relative real-time ordering of operations. Indeed, in an asynchronous system, the only ordering information available to them is based on the fact that sending a message precedes its delivery. We show that as a result, message chains must be used extensively to ensure linearizability. 
This paper studies the communication requirements of linearizable implementations of atomic registers in asynchronous message passing systems. 
We start by proving two general theorems that relate message chains to the ability to delay and reorder actions and operations in an execution of an asynchronous system, without the changes being noticeable to the processes. These are then used to prove that linearizable register implementations must create extensive message chains among operations of \emph{all} types. In particular, our results imply  that linearizable implementations in asynchronous systems are necessarily  costly and nontrivial, and provide insight into their structure. 
%
\end{abstract}
\section{Introduction}

Linearizability \cite{HerlihyLineari} is a fundamental correctness criterion and is the gold standard for concurrent implementations of shared objects. Informally, an object implementation is linearizable if in each one of its executions, operations appear to occur instantaneously, in a way that is consistent with the execution and the object's specification. Linearizable implementations have been developed for a variety of concurrent objects 
 \cite{afek1993atomic,michael1996simple,herlihy2020art} and is also widely used in the context of {\em state-machine replication} (SMR) mechanisms \cite{SMR1990Schneider,SMReurosys2020,SwiftPaxos}.
 Understanding the costs that linearizable implementations imply and optimizing their performance is thus crucial.  Lower bounds on linearizable implementations are rare in the literature. Our paper makes a significant step towards capturing inherent costs of linearizability in the important case of linearizable register implementations, and provides a new formal tool for capturing the necessary structure of communication in register implementations.
 
In an execution of a linearizable implementation, the actions performed and values observed by processes  depend on the real-time ordering of non-overlapping operations \cite{HerlihyLineari}. However,  processes do not have direct access to real time in the asynchronous setting, and this makes  satisfying linearizability especially challenging.    
 The only way processes can obtain information about the real-time order of events in asynchronous message-passing systems is via {\em message chains} ({\em cf.}\ Lamport's {\em happens before}  relation \cite{Lam78causal}). 
Roughly speaking, a message chain connects process $i$ at (real) time $t$ and process $j\neq i$ at $t'$ if 
there is a sequence of messages starting with a message sent by~$i$ at or after~$t$, ending with a message received by~$j$ no later than time~$t'$, such that every message is received by the sender of the following message in the sequence, before the following message is sent.\footnote{A formal definition appears in \Cref{sec:model2.2}.} Message chains can be used to ensure the relative real-time order of events. 
Moreover, as we formally show, in the absence of a message chain relating events at distinct processes, there can be no way to tell what their real-time order is. 
  This paper establishes the central role that message chains must play in achieving linearizability in an asynchronous system.

 Registers constitute a central abstraction in distributed computing. In their seminal paper \cite{ABD}, Attiya, Bar-Noy and Dolev  provide a linearizable implementation of single-writer multi-reader (SWMR) registers in an asynchronous message passing model where processes are prone to crash failures. This implementation was extended to the multi-writer multi-reader (MWMR) case in~\cite{MWMRLynch}. Since then, there has been significant interest in implementing registers in asynchronous message passing models. 
 In \cite{ABD}, quorum systems are used to guarantee a message chain between {\em every pair} of non-overlapping operations. This is costly, of course, both in communication and in execution time. Is it necessary?
 
In a linearizable implementation of a MWMR  register, every process can issue reads and writes, and a read should intuitively return the most recent value written. It is to be expected that a reader must be able to access previous write operations, and especially the one whose value its read operation returns. But should writing a new value, for example, require message chains from all previous reads and writes? Must a process that has read a value communicate this fact to others?
 Interestingly, we show in this work that typically, the answer is {\em yes}. Moreover, we prove that 
 every operation of a {\em fault-tolerant} implementation of a MWMR register must communicate with a quorum set before it~completes.
 \\
 
  The main contributions of this paper are 
  \begin{enumerate}
   \item We show that in a linearizable implementation of a register in an asynchronous setting, every operation, regardless of type, might need to have a message chain to arbitrary operations in the future. Moreover, in an $f$-resilient implementation, before a process can complete an operation, it must construct a round-trip message chain interaction with nodes in a quorum set of size greater than~$f$. 
      These requirements apply to {\em every} execution and thus, provide a natural way for establishing lower bounds on the performance of register implementations and related applications not only in the worst case, but also in optimistic executions (a.k.a.\ fast paths) \cite{SwiftPaxos,FastSlow,fastAtomicRegister04}.
      We expect this work to serve as a tool for analyzing the efficiencies of existing implementations and also as a guide for implementing new linearizable objects in the future.
      
      \item We show these results by formulating and proving two useful and general theorems about coordination in asynchronous systems. One relates message chains to the ability to delay particular actions in an execution of an asynchronous system for an arbitrary amount of time, without the delay being noticeable to any process in the system. The other relates them to the ability to change the relative real-time order of operations on concurrent objects in manners that may cause violations of linearizability requirements.

     \end{enumerate}
    Interestingly, a significant amount of communication in a linearizable implementation is required for timing purposes, rather than for transferring information about data values. 
    Our results apply verbatim if message passing is replaced by communication via asynchronous single-writer single-reader (SWSR) registers or in hybrid models (\cite{HybridNaama}, \cite{hadzilacos2022atomic}) for a suitably modified notion of message chains.  They also extend to other variants of linearizability, such as strict linearizability \cite{aguilera2003strict}.

This paper is structured as follows:
\Cref{sec:related} presents related work.
In \Cref{sec:model} we present the model and preliminary definitions and results about message chains, real time ordering and the local equivalence of runs. In \Cref{sec:pastlemma} we prove a theorem about the ability to delaying actions in a way that processes cannot notice. This is used in \Cref{sec:operations} to show that certain operations 
can be reordered in a run, in a similar fashion. These results, which can be applied to 
arbitrary objects, are next used for the study of atomic register implementations.  \Cref{sec:registers} contains definitions of registers  and linearization in our setting.  \Cref{sec:XaYb}  provides 
 general results showing the need for  message chains between operations in executions of linearizable register implementations. In  \Cref{sec:quorumsandfailures} we show how the presence of failures combined with the  results of the previous sections imply the necessity of using quorum systems.

\section{Related Work}\label{sec:related}
Attiya, Bar-Noy and Dolev's paper (ABD) \cite{ABD} shows
how to implement shared memory via message passing in an asynchronous message passing model where processes are prone to crash failures. Their algorithm (which we shall call ABD) is $f$-resilient and makes use of quorum systems. Each write or read operation performs two communication rounds. In each communication round by $p$, process $p$ sends messages to all $n$ processes and waits for replies  from $n-f$ processes before it proceeds to the next communication round.

In \cite{fastAtomicRegister04} and \cite{SemiFastRegPODC20}, Dutta \emph{et al.}\ and Huang \emph{et al.}, respectively, consider a model consisting of disjoint sets of servers, writers and readers and where at least one process can fail ($f\geq 1$). 
They study implementations of an atomic register where  read or write operations are {\em fast}, by which they mean that the operations terminate following a single communication round.
In \cite{fastAtomicRegister04}, an SWMR register implementation is provided with fast  reads and writes, if the number of readers is small enough relative to the number of servers and the maximal number of failures.
They also  prove that MWMR register implementations with both fast read and fast writes are impossible. \cite{SemiFastRegPODC20}  
proves that 
implementations with fast writes are impossible and by showing under which conditions (on the number of failures) implementations with fast reads exist. The models of \cite{fastAtomicRegister04,SemiFastRegPODC20} assume crash failures. Our results in  Sections~\ref{sec:pastlemma}-\ref{sec:XaYb} are valid both when processes are guaranteed to be reliable (no failures) and in the presence of crash failures.

In \cite{naserpastoriza2023} Naser-Pastoriza {\em et al.\/} consider networks where channels may disconnect. As one of that paper's main contributions, it establishes minimal connectivity requirements for linearizable implementations of registers in a crash fault environment where channels can drop messages. Informally, it is shown that (1) all processes where obstruction-freedom holds must be strongly connected via correct
channels; and (2). If the implementation tolerates $k$ process crashes and $n=2k+1$, then any process where obstruction-freedom holds must belong to a set of more than $k$ correct processes strongly connected by correct channels. 
 
 The works \cite{QuorumDetector2009,delporte2004weakest, delporte2008weakest} show  that {\em quorum} failure detectors, introduced by Delporte-Gallet \emph{et al.} in~\cite{delporte2004weakest}, are the weakest failure detectors enabling the implementation of an atomic register object in asynchronous message passing systems. This  class of failure detectors capture the minimal information regarding failures that processes must possess to in linearizable implementation of registers.

Variants of linearizability that differ in the way crashes are handled  have been defined in the context of NVRAMs; see  Ben-David \emph{et.~al} \cite{ben2022survey} for a survey.  Another important
variant is {\em strong} linearizability,  introduced by Golab \emph{et.~al} in \cite{Golab2011Strong}. They showed that in a randomized algorithm, executions behave exactly as if using atomic objects if and only if the implementation is  strongly linearizable. Attiya \emph{et.~al} \cite{AEW2021} proved that multi-writer registers do not have strongly linearizable nonblocking implementations in message-passing systems. In \cite{hadzilacos2022atomic}, Hadzilacos \emph{et.~al} showed that the ABD implementation of an atomic SWMR register in message passing systems is not strongly linearizable. Finally, Chan \emph{et.~al} \cite{Chan21} proved that single-writer registers do not have strongly linearizable nonblocking implementations in message-passing systems. 
Our results in \Cref{sec:pastlemma,sec:operations,sec:registers} are valid in asynchronous models in general and can thus also be used in the analysis and the study of such variants of linearizability.




\section{Model and Preliminary Definitions} \label{sec:model}
\subsection{Model}
While asynchronous systems are often captured using an interleaving model, 
we adopt the asynchronous message passing model from \cite{FHMV03}, in which several events can take places at the same time. This facilitates reasoning about the time at which actions and operations occur, and analyzing the possibility of modifying the timing of some operations while leaving the timing of other operations unchanged.   
We briefly describe the model here and refer the reader to \Cref{sec:detailedmodel} for the complete detailed model. The detailed model is required mainly for the proof of \Cref{thm:paslemma2.1} which is lays the technical basis for most of our analysis.  

We consider an asynchronous message passing model with $n$ processes, connected by a communicated network, modeled by a directed graph where an edge from process~$i$ to process~$j$ is called a {\em channel}, and denoted by $\chan{i,j}$. The environment, which plays the role of the adversary, is in charge of scheduling processes, of delivering messages, and of invoking operations (such as reads, writes etc.) at a process. 
A run  of the system is an infinite sequence $r\,=\,r(0), r(1),\ldots$ of global states, where each global state $r(m)$ determines a local state for each process, denoted by $r_i(m)$. We identify time with the natural numbers, and consider $r(m)$ to be the system's state at time~$m$ in~$r$.  For ease of exposition, we assume that messages along a channel are delivered in FIFO order. Moreover, we assume that the local state of a process~$i$  keeps track of of the events it has been involved in so far: all actions it has performed, all messages it sent and received, and all operations invoked at~$i$, up to the current time. 
Asynchrony of the system is captured by assuming that processes moves, message deliveries  and operation invocation are scheduled  in an arbitrary nondeterministic order. Thus messages can take any amount of time to be delivered, and processes can refrain from performing moves for arbitrarily long time intervals. 
We consider actions to be performed in {\em rounds}, where round~$m$ occurs between time~$m$ and time~$m+1$. The transition from $r(m)$ to $r(m+1)$ is based on the actions performed by the environment and by all processes that move in round $m+1$. 

A process $i$ is said to be correct in $r$ if it is allowed to move (by the environment)  infinitely often in~$r$. Otherwise process $i$ is {\em faulty} (or \emph{crashes}) in $r$. 
We say that a message $\mu$ is {\em lost} in~$r$ if it is sent in~$r$ and never delivered. 
A system is said to be \emph{reliable} if no process ever fails and no message is ever lost, in any of its runs. Finally, a protocol is said to be \emph{$f$-resilient} if it acts correctly in all runs in which no more than~$f$ processes are faulty.

\begin{toappendix}
\section{Detailed Model}\label{sec:detailedmodel}

Our model is based on the asynchronous message passing model of Fagin \emph{et al.} [9].   We consider  a set $\Pi$ of $n$ processes.  connected via a communication network $Net$, defined by a directed graph $(\Pi,E)$.  Every edge $(i,j)\in E$  is associated with a {\em channel}, and denoted by $\chan{ij}$, consisting of a set of records of the form $|\mu, t|$. Such a record represents the fact that~$\mu$ was sent by $i$ to $j$ in round $t$ and is still in transit (i.e., it has not been delivered yet). In addition to the processes, the {\em environment} (denoted by $e$) models what is commonly referred to as the adversary. 

\begin{itemize}
    \item \textbf{Actions}: For $i\in \Pi$, the set of actions $Act_i$ it can perform consists of local actions $\alpha_i$ (possibly including a $no\_op_i$ action) and message send actions of the form $send_i(\mu,j)$. 
The environment actions will play the role of determining when messages are delivered, when processes are scheduled to move and external inputs to individual processes. 
Thus, an environment action $\alpha_e$ is a tuple $\vec{\eta}=\langle \eta_1,...,\eta_n\rangle$ containing a component $\eta_i$ for every process $i$.  
Each $\eta_i$ is either $\move_i$, $\skipp_i$, $\inv_i(x)$, or $\del_i(|\mu, t|, j)$ for some message~$\mu$ and process $j \neq i$. The $\move_i$ action means that $i$ performs an action according to its protocol as defined below;  $\skipp_i$ means that~$i$ is ignored in the current round; $\inv_i(x)$ means that~$i$ will receive the external input~$x$, while $\del_i(|\mu, t|, j)$ means that~$i$ will  receive the message $\mu$ from~$j$, provided that this message is in transit, and is the next message to be delivered in FIFO order. 
\item \textbf{States}: 
The local state of a process at any given point is its local event history $h_i$, containing an initial value and a sequence of local events: all the messages $i$ received, external inputs and all the actions $i$ performed, all arranged in the order in which~$i$ observed them.  The local state of the environment contains the complete record of all actions performed so far, as well as the current contents of $\chan{ij}$ for all network edges $(i,j)\in E$. 
A global state is a tuple $g=(\ell_e,\ell_1,\ldots,\ell_n)$ containing a state for the environment and a state for each one of the processes. An {\em initial} global state is one in which all local states contain only the initial values.
\item \textbf{Protocols}: 
A {\em protocol} $P_i$  
associates a nonempty set of actions (of $Act_i$) with every local state of the process~$i$.
If $P_i(\ell_i) = S$, then the action performed by $i$ when scheduled to move in state $\ell_i$ will be one of the elements of $S$. A protocol for the processes  has the form $P = (P_1,\dots, P_n)$, and it associates a protocol $P_i$ with every process $i\in \Pi$.

\item \textbf{Environment Protocol}: The environment’s protocol, which we denote by $P_e^a$, is given by $P_e^a(\ell_e)\triangleq \{\vec{\eta}:\vec{\eta}\in Act_e\}$. In words, $P_e^a(\ell_e)$ performs, for every process $i \in P$, an independent, nondeterministic choice of $\eta_i$ among the possibilities of $\move_i$, $\skipp_i$, $\del_i(\cdot, \cdot)$ or $\inv_{i}(\cdot)$.
 \item \textbf{Transition Function}: A  joint action is a tuple $(\vec{\eta},\alpha_1,\ldots,\alpha_n)$ with $\vec{\eta}\in Act_e$ and $\alpha_i\in Act_i$ for each $i\in \Pi$. The transition function modifies the environment’s local state~$\ell_e$ by appending the current round’s joint action at the end of the event history $h$. Local states are transformed as follows: If $\eta_i=\move_i$ then the action $\alpha_i\in P_i(\ell_i)$ that $i$ performs (as recorded in the joint action added to $h$) is appended at the end of~$h_i$. 
Moreover, if $\eta_i=\move_i$ and $i$’s action is $send_i(\mu,j)$ where $(i,j)$ is a link in $Net$, then a record $|\mu,m|$, where $m$ is the current (sending) round, is added to $\chan{ij}$. 
    Similarly, if $\eta_i=\inv_i(\cdot)$, then 
    this external input is appended at the end of $h_i$. If~$\eta_i=\del_i(|\mu,t|,j)$ and $|\mu,t|$ is the oldest message in~$\chan{ji}(m)$ (the message~$\mu$ was sent in round~$t$ by $j$, is still in transit at the current time~$m$, and is the next message to be delivered according to FIFO order), then this record $|\mu,t|$ is removed from $\chan{ji}$ and $(j,\mu)$ is appended to the end of~$h_i$. 
    In this case, $\mu$ is said to be delivered in round $m$ in $r$. The local state $\ell_i$ of $i \in P$ remains unchanged if $\eta_i=\skipp_i$ or if $\eta_i=\del_i(|\mu, t|, j)$ and $|\mu, t| \notin \chan{ji}(m)$, and the $i$-component in the joint action~$\vec{\alpha}_m$ performed at time $m$ is ‘$\bot_i$’.
 \item \textbf{Runs}: A run $r$ is an infinite sequence of global states, whose first element $r(0)$ is the initial global state and we use $r(m)$ to denote the ($m+1$)th state in the sequence. We identify time with the natural numbers, and think of $r(m)$ as being  the global state at time $m$ in~$r$. We denote by $r_i(m)$ the local state of process $i$ in $r(m)$.

\item \textbf{Run of a protocol~$P$}:    A run $r$ is called a run of $P$ if \begin{itemize}
    \item[(i)] $r(0)$ is an initial global state, and 
    \item[(ii)] for every $m\ge 0$, 
there is a joint action $\vec{\alpha}=(\vec{\eta},\alpha_1,\ldots,\alpha_n)$ with $\vec{\eta}\in Act_e$ and $\alpha_i\in P_i(r_i(m))$ for every $i=1,\ldots,n$ such that $r(m+1)$ is obtained  by applying the transition function to $\vec{\alpha}$ and $r(m)$.  
\end{itemize}
\end{itemize}
\vspace{-1mm}
A protocol $P_i$ is deterministic if it always specifies a unique action, i.e., if $P_i(\ell_i)=S$ then $|S|=1$. We remark that while processes may or may not follow a deterministic protocol, the environment's protocol is highly nondeterministic. The asynchronous aspect of an a.m.p. 
is captured mainly by the environment’s protocol and the transition function: The environment can delay the delivery of a message for arbitrarily long, and it can similarly delay a process from taking a step, and this is independent of how many steps others take, and of whether they receive messages sent to them. Moreover, the transition function is such that a process’ local state changes only if the process either receives a message, takes a step, or that the environment's action is an invocation. Thus, it has no way of telling whether and how much time has passed since its last move. 

\noindent{\bf Crashes and loss of messages.} A process $i$ is said to be correct in $r$ if it is allowed to move ($\eta_i=\move_i$)  infinitely often in~$r$. Otherwise process $i$ is {\em faulty} (or \emph{crashes}) in $r$. 
We say that a message $\mu$ is {\em lost} in~$r$ if it is sent in~$r$ and never delivered.

A system is said to be \emph{reliable} if no process ever fails and no message is ever lost, in any of its runs. A protocol is said to be \emph{$f$-resilient} if it acts correctly in all runs in which no more than~$f$ processes are faulty. 
\end{toappendix}
\subsection{Message Chains,  Real-time Ordering and Local Equivalence}\label{sec:model2.2}
As stated in the introduction, the real-time order of events in a system plays a central role in linearizable protocols. The main source of information about the order of events in asynchronous systems are message chains. 
We denote by $\theta=\node{p,t}$ a {\bf process-time} pair (or a {\em node}) consisting of the process $p$ and time $t$.
    Such a pair is used to refer to the point on~$p$'s timeline at real time~$t$.
We can inductively define a message chain between nodes of a given run as follows. 


\begin{definition} [Message chains] \label{def:msg-chain}
    There is {\bf a message chain} from~${\theta}=\node{p,t}$ to 
 ${\theta'}=\node{q,t'}$ in a run~$r$, denoted by ${\theta\mc_{r}\theta'}$, if
 \begin{itemize}
     \item[(1a)] $p=q$ and $t<t'$, 
     \item[(1b)] $p$ sends a message to~$q$ in round $t+1$ of~$r$, which arrives no later than in round $t'$, or
     \item[(2)] there exists $\theta''$ such that $\theta\mc_{r}\theta''$ and $\theta''\mc_{r}\theta'$.
  \end{itemize}
\end{definition}

Lamport calls `$\mc_r$' the \emph{happens before} relation \cite{Lam78causal}. As we now show, the existence of message chains indeed implies real-time ordering. 
We write $\theta<_r\theta'$ if~$\theta=\node{p,t}$ and~$\theta'=\{q,t'\}$ are nodes in~$r$ and $t<t'$. An immediate implication of \Cref{def:msg-chain} is 
\begin{observation}\label{obs: mc and rt}
    If $\theta\mc_r\theta'$ then $\theta<_r\theta'$.
\end{observation}
\begin{proof}
   Let $\theta=\node{p,t}$ and  
 $\theta'=\node{q,t'}$. The proof is by induction on the minimal number of applications of step (2) in \Cref{def:msg-chain} needed to establish that $\theta\mc_r\theta'$.  If $\theta\mc_r\theta'$ by (1a) then $t<t'$. Similarly, if it is by (1b), then $t<t'$ because a message sent in round~$t+1$ can only arrive in a round $t'\ge t+1>t$. Finally, if $\theta\mc_r\theta'$ by clause (2), then for some node~$\theta''=\node{p'',t''}$ we have that $\theta\mc_{r}\theta''$ and $\theta''\mc_{r}\theta'$, where, inductively, $t<t''$ and $t''<t$. It follows that $t<t'$, as required. 
\end{proof}
The converse is not true: It is possible for~$\theta$ to appear before~$\theta'$ in real time, without a message chain between them. As we shall see, however, in the absence of a message chain, processes will not be able to detect the ordering between the nodes.


Roughly speaking, the information available to a process at a given point is determined by its local state there. A process is unable to distinguish between runs in which it passes through the same sequence of local states. 
We will find it useful to consider when two runs cannot \emph{ever} be distinguished by any of the processes. 
Formally:
\begin{definition}[Local Equivalence]
    Two runs $r$ and $r'$ are called \emph{locally equivalent}, denoted by  $r\loc r'$,  if for every process $j$, a local state $\ell_j$ of~$j$ appears in~$r$ iff $\ell_j$ appears in~$r'$. 
\end{definition}
Recall that the local state of a process~$i$ consists of its local history so far. Consequently, an equivalent definition of local equivalence is that if two runs are locally equivalent, then every process starts in the same state, performs the same actions and sends and receives the same messages, all in the same order, in both runs.

A node $\theta=\node{i,t}$ of~$i$ in~$r$ is said to \emph{correspond} to node $\theta'=\node{j,t'}$ of $r'$,  denoted by $\theta\sim\theta'$, if $i=j$ (they refer to the same process) and the process has the same local state at both (i.e., $r_i(r)=r'_i(t')$). We will make use of the following properties of local equivalence (the proof of \Cref{lem: loc_eq} appears in the Appendix): 
\begin{restatable}{lemma}{primelemma}\label{lem: loc_eq}
    Let $r$ and $r'$ be two runs such that $r\loc r'$. Then  
\begin{enumerate}[(i)]
    \item If $\theta_1\mc_r\theta_2$ then $\theta_1'\mc_{r'}\theta_2'$ holds for all nodes $\theta_1'$ and $\theta_2'$  of~$r'$ such that  $\theta_1\sim\theta'_1$ and~$\theta_2\sim\theta'_2$ 
    \item If $r$ is a run of protocol $P$, then $r'$ is also a run of $P$
    \item A process~$i$ fails in~$r$ iff it fails in~$r'$, and similarly
     \item A message~$\mu$ is lost in~$r$ iff the same message is lost in~$r'$
\end{enumerate}
\end{restatable}
\begin{toappendix}
\\\\
We now restate and prove \Cref{lem: loc_eq}.

\primelemma*

\begin{proof}
We prove each one of the claims.
\begin{enumerate}[(i)]
    \item
    Let $\theta_1,\theta_2$ in $r$ such that $\theta_1\mc_r\theta_2$. 
     Denote by $\theta_1=\alpha_1,\alpha_2,\dots,\alpha_k=\theta_2$ the nodes constituting this message chain such that $\alpha_{i+1}$ is obtained from $\alpha_{i}$ applying (1a) or (1b) of \Cref{def:msg-chain}.
    We prove by induction on~$k$ that $\alpha_k$ has a corresponding message chain $\alpha_1'\mc_{r'}\alpha_k'$ in $r'$ when $\alpha_1\sim\alpha_1'$ and $\alpha_k\sim\alpha_k'$.
    
    Base: $k=1$. The base case results directly from the fact that every local state in $r$ appears in $r'$ and vice-versa. Thus there is $\alpha_1'\sim\alpha_1$ in $r'$.

    Step: Let $k>1$ and assume inductively that the claim holds for $\theta_1=\alpha_1,\alpha_2,\ldots,\alpha_{k-1}$.  If $\alpha_{k}$ is obtained from $\alpha_{k-1}$ by (1a) of \Cref{def:msg-chain}, then let $\alpha'_{k-1}=\node{p,t_{k-1}}$ be the node of~$r'$ such that $\alpha_{k-1}\sim\alpha'_{k-1}$. By local equivalence between~$r$ and~$r'$ there must be a node $\alpha'_k\sim\alpha_k$ of~$p$ in~$r'$, and the claim holds. 
    
     If $\alpha_{k}$ is obtained from $\alpha_{k-1}$ by (1b), meaning that a message is sent at node $\alpha_{k-1}$ and arrives no later than at $\alpha_{k}=\node{p,t_k}$, then since the send and the delivery of messages are registered in processes local states, we have by definition of locally equivalence that this message is also sent at $\alpha'_{k-1}$ and arrives no later than a node  $\alpha'_{k}\sim\alpha_k$ of~$p$ in~$r'$.

    \item Assume $r$ is a run of $P$. We show that if action $\alpha_i$ is performed in $r'$ then it is an action of $P$. Let $\alpha_i'$ be an action performed by $i$ in $r'$ and denote by $l_i$ the state of $i$ right after performing this operation . By definition of local equivalence, there is a point in $r$ such that $i$ has local state $l_i$ in $r$. I.e., $i$ performed action $\alpha_i'$ in $r$, which by assumption is a run of $P$.
    \item If a process $i$ fails in $r$, then there is a time $t$ from which the environment action of $i$ is not $\move_i$ anymore. Thus, there is a finite number of actions registered to its local state along the run. By definition of local equivalence, this is the case also for $r'$, i.e., $i$ fails also in $r'$.
    \item Let $\mu$ be a message sent by $i$ to $j$ at node $\theta$ in $r$. It follows from item (i) that $\mu$ is sent at some node $\theta'\sim\theta$ in $r'$. If $\mu$ is lost in $r$, then this message is never delivered to process $j$ and thus this reception is never added to the local state of $j$. By item (i) it follows $\mu$ is also lost in $r'$.
\end{enumerate}
    
\end{proof}
\end{toappendix}




\section{Delaying the Future while Maintaining the Past}\label{sec:pastlemma}


We are now ready to state and prove the main theorem that will allow us to capture the subtle interaction between message chains and the ability to reorder operations in an asynchronous system. 
\begin{definition}[The past of~$\theta$]\label{def:past}
For a node~$\theta$ in a run~$r$, we define ~$\,\past_r(\theta)~\triangleq~\{\theta'\,|\;\theta'\mc_r\theta\}$.  
\end{definition}

 Chandy and Misra have already shown that, in a precise sense, in an asynchronous system, a process at a given node cannot know about the occurrence of any events   except for ones that appear in its past \cite{ChM}.  Our theorem will show that for any given node $\theta$ in a run~$r$ (which we think of as a ``pivot node'') {\bf all} events that occur outside $\past_r(\theta)$ can be pushed into the future by an arbitrary amount~$\Delta>0$, without any node observing the change. 

\begin{figure}[H]
        \centering
        \includegraphics[width=14cm]{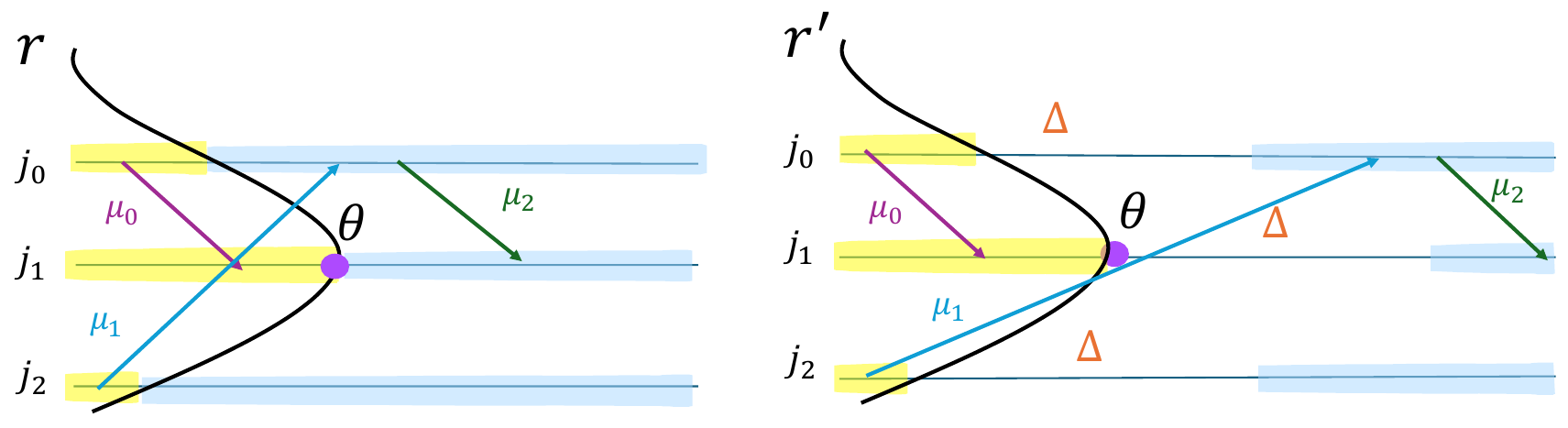}
        \caption{Delaying events by~$\Delta$ relative to the past of a node~$\theta$ (the ``pivot'').}
        \label{fig:pastlemma}
    \end{figure}

\begin{theorem}[Delaying the future]\label{thm:paslemma2.1}
 
    Fix a run~$r$ of a protocol~$P$, a node $\theta=\node{i,t}$, and a delay $\Delta>0$. For each process~$j$  denote by $t_j$ the minimal time $l\geq 0$ such that $\node{j,l}\not\mc_r\theta$ (i.e., $\node{j,t_j}$ is the first point of~$j$ that is not in the past of~$\theta$ in~$r$). Then there exists a run~$r'\approx r$  
    satisfying, for every process~$j$:
                  \begin{equation*}
    	{r}_j(m)=
    	\begin{cases}    		
    		{r'}_j(m)& \textrm{for all~}m\leq t_j\\
    	{r'}_j(m+\Delta) & \textrm{for all~}m\geq t_j+1
    	    	\end{cases}     
    \end{equation*}
    
        
    \end{theorem}

This theorem lays the technical foundation for most of our analysis in this paper. We start by providing a sketch of its proof, and follow with the full proof. 
\begin{proofsketch}  
Recall that we are given~$r$, $\theta$ and $\Delta$. For every process~$j$ there is an earliest time $t_j$ such that $\node{j,t_j}\notin\past_r(\theta)$.
 We now construct a run~$r'$ that agrees with $r$ on all nodes of~$\past_r(\theta)$. I.e., for every node $\theta'=\node{p,t'}\in\past_r(\theta)$,  then the same actions occur in round~$t'$ on~$p$'s timeline, and $r_p(t')=r'_p(t')$. 
 Moreover, outside of $\past_r(\theta)$ the run~$r'$ is defined as follows. The 
 environment in~$r'$ ``puts to sleep'' every process~$j$ (by performing $\skipp_j$ 
 actions)  for a duration of $\Delta$ rounds starting from round $t_j+1$ and ending in round~$t_j+\Delta$. Every message that, in~$r$, is delivered to~$j$ at a round $m>t_j$ is  delivered $\Delta$ rounds later, i.e., in round $m+\Delta$,  in~$r'$. Similarly, every message sent by~$i$ after time~$t_i$ in~$r$ is sent $\Delta$ rounds later in~$r'$. 
 A crucial property of this construction is that, by definition of~$\mc_r$, if the sending of a message is delayed by~$\Delta$ in~$r'$---the sending node is {\em not} in $\past_r(\theta)$---then its delivery is delayed by~$\Delta$ as well.  Consequently, every message sent in~$r'$ is delivered at a time that is greater than the time it is sent, and so~$r'$ is a legal run.  
 What remains is to check that the run~$r'$ is indeed locally equivalent to~$r$. This careful and somewhat tedious task is performed in the full proof that follows below.
\end{proofsketch}
 
 As illustrated in \Cref{fig:pastlemma}, the run~$r'$ contains a band of inactivity that is~$\Delta$ rounds deep in front of the boundary of $\past_r(\theta)$. Since $\Delta$ can be chosen arbitrarily, \Cref{thm:paslemma2.1} can be used to rearrange any activity that does not involve nodes of $\past_r(\theta)$, even events that may be very early, to occur strictly after~$\theta$ in~$r'$. 
 Crucially, no process is ever able to distinguish among the two runs. 
    
    \begin{proof}[Proof of \Cref{thm:paslemma2.1}]

    To simplify the case analysis in our proof, we define 
    \begin{equation*}
    	\shift{m,t_j} ~\triangleq~
    	\begin{cases}    		
    		m& m\leq t_j\\
    	m+\Delta & m\geq t_j+1
    	    	\end{cases}     
    \end{equation*}
    Notice that the range of $\shift{m,t_j}$ for $m\ge 0$ is the set of times~$m'$ 
     not in the interval $t_j+1\le m'\le t_j+\Delta$.  
     Moreover, observe that $\shift{m-1,t_j}=\shift{m,t_j}-1$ for all $m>0$ such that $m\ne t_j+1$. 
     We shall construct a run ${r'}\approx r$ satisfying, for every process~$j$ and all $m\ge 0$:
                     
                      \begin{itemize}
             
                  \item[(i)]  $r_j(m)={r'}_j(\shift{m,t_j})$ for all $m\ge 0$, and 
                     
                    \item[(ii)] Process~$j$ performs the same actions and receives the same messages in round $m$ of~$r$ and in round $\shift{m,t_j}$ of~${r'}$, for all $m\ge 1$.
        \end{itemize}

We construct ${r'}$ as follows. Both runs start in the same initial state: ${r'}(0)=r(0)$. 
Denote the environment’s action in $r$ in round $m$ by $\eta(r,m)=(\eta_1(r,m),\dots,\eta_n(r,m))$. 
For every process~$j$ the environment's actions $\eta_j$ satisfies  $\eta_j({r'},m')\triangleq\skipp_j$  for all $m'$ in the range $t_j+1\le m'\leq t_j+\Delta$. 
For all $m\ge 0$  we define
\begin{equation*}
	\eta_j({r'},\shift{m,t_j}) ~\triangleq~
	\begin{cases}
		\eta_j(r,m) & ~\hspace{-10mm}\textrm{if~} \eta_j(r,m)\in\{\skipp_j,\move_j,\inv_j(x)\}\\
		\del_j(|\mu,\shift{m_h,t_h}|,h) & \mathrm{if~} \eta_j(r,m)=\del_j(|\mu,m_h|,h)
	\end{cases}       
\end{equation*}

 As for process actions, for all~$j$ and $m> 0$, if $\eta_j({r'},\shift{m,t_j})=\move_j$ and ${r'}_j(m-1)=r_j(m-1)$ then~$j$ performs the same action $\alpha_j\in P_j(r_j(m-1))$ in round $\shift{m,t_j}$ of~${r'}$ as  in round~$m$ of~$r$, and otherwise it performs an arbitrary action from \mbox{$P_j({r'}_j(\shift{m-1,t_j})$} in round $\shift{m,t_j}$ of~${r'}$. 
 Notice that, by definition, all processes follow the protocol~$P=(P_1,\ldots,P_n)$ in~${r'}$.
Moreover, observe the following useful property of~${r'}$:  

\noindent\underline{\tt Claim 1}:\quad ${r'}_j(\shift{m,t_j}-1)={r'}_j(\shift{m-1,t_j})$ for all $m>0$.
\begin{proof}
      We consider two cases: 
\begin{itemize}
\item $m=t_j+1$:\quad Observe that ${r'}_j(t_j+\Delta)={r'}(t_j+\Delta-1)=\dots={r'}_j(t_j)$ since by definition of the run~${r'}$, we have that  $\eta_j({r'},m')=\skipp_j$ for all $t_j+1\leq m'\leq t_j+\Delta$. So, ${r'}_j(\shift{m,t_j}-1)={r'}_j(\shift{t_j+1,t_j}-1)={r'}_j(t_j+1+\Delta-1)={r'}_j(t_j+\Delta)={r'}_j(t_j)={r'}_j(\shift{m-1,t_j})$.
    \item $0<m\ne t_j+1$: 
    If $m\leq t_j$ then by definition of  $\mathtt{shift}_\Delta$ we have  that  \mbox{$\shift{m,t_j}=m$} and $\shift{m-1,t_j}=m-1=\shift{m,t_j}-1$. 
    Similarly, if $m>t_j+1$ then $\shift{m,t_j}=m+\Delta$ and $\shift{m-1,t_j}=m-1+\Delta=\shift{m,t_j}-1$. In both cases we obtain that ${r'}_j(\shift{m,t_j}-1)={r'}_j(\shift{m-1,t_j})$, as desired. \vspace{-3mm}
\end{itemize}
\end{proof}
 
 We are now ready to prove that~${r'}$ is a legal run of $P$ satisfying (i) and (ii). We prove this by induction on~$m\ge 0$, for all processes~$j$. 

 \textbf{Base,} $m=0$: By definition of ${r'}$ we have that ${r'}_j(0)=r_j(0)$.

 \textbf{Step,} $m>0$: Assume inductively that (i) and (ii) hold for all processes~$h$ at all times strictly smaller than~$m$. 
We start by establishing:

\noindent\underline{\tt Claim 2}:\quad
If a message~$\mu$ sent by a process~$h$ at time $m_h$ is delivered to~$j$~~in round~$m$ of~$r$, then~$|\mu,\shift{m_h,t_h}|\in chan_{hj}$ at time $\shift{m,t_j}-1$ of~${r'}$.
\begin{proof}
    Clearly, if ~$\mu$ is delivered to~$j$ in round~$m$ of~$r$ then $\eta_j(r,m)=\del_j(|\mu,m_h|,h)$ for some process $h\ne j$ and round $m_h<m$. 
    By the inductive assumption for $h$ and $m_h<m$, we have that $\mu$ is sent in round $\shift{m_h,t_h}$ of~${r'}$. In addition, by definition of~${r'}$, for all $m'<\shift{m,t_j}$ it holds that $\eta_j({r'},m')\neq \del_j(|\mu,\shift{m_h,t_h},h)$. So $|\mu,\shift{m_h,t_h}|\in chan_{hj}$ at time $\shift{m,t_j}-1$ in~${r'}$.
     
\end{proof}

Recall that we have by the inductive assumption that ${r'}_j(\shift{m-1,t_j})=r_j(m-1)$. Claim 1 thus implies that
\begin{equation}
   ~~~ {r'}_j(\shift{m,t_j}-1)~=~r_j(m-1).
\end{equation}\label{eq:induc}

We can now show that (i) and (ii) hold for~$j$ and~$m$ by cases depending on the environment's action $\eta_j(r,m)$ in round $m$ of $r$: 
\begin{itemize}
	\item[-] $\eta_j(r,m)=\skipp_j$:\quad By definition of $\eta_j$ for ${r'}$, we have that $\eta_j({r'},\shift{m})=\skipp_j$. So, ${r'}_j(\shift{m,t_j})={r'}_j(\shift{m,t_j}-1)=r_j(m-1)$, proving (i).
    Moreover, no action is performed by $j$ neither in $r$ nor in ${r'}$ and no message is delivered to~$j$ in either case, ensuring that (ii) also holds.
    \item[-] $\eta_j(r,m)=\inv_j(x)$: \quad In this case, $\eta_j({r'},\shift{m})=\inv_j(x)$, implying that ${r'}_j(\shift{m,t_j})=r_j(m)$.
	\item[-] $\eta_j(r,m)=\move_j$:\quad In this case, $\eta_j({r'},\shift{m})=\move_j$  by definition of $\eta_j$ for ${r'}$. By (\ref{eq:induc}) we have that ${r'}_j(\shift{m,t_j}-1)=r_j(m-1)$. So by definition of ${r'}$, process~$j$ performs the same action $\alpha_j\in P_j(r_j(m))$ in the round $\shift{m,t_j}$ of ${r'}$ as it does in the round $m$ of $r$. This also ensures ${r'}_j(\shift{m,t_j})=r_j(m)$. In addition, no message is delivered  in round $m$ of~$r$ and none is delivered to it in round~$\shift{m,t_j}$ of ${r'}$.

  \item[-] $\eta_j(r,m)=\del_j(|\mu,m_h|,h)$:\quad In this case, no action is performed by $j$. By definition,  $\eta_j({r'},\shift{m,t_j})=\del_j(|\mu,\shift{m_h,t_h}|,h)$. Recall that by (\ref{eq:induc}) we have ${r'}_j(\shift{m,t_j}-1)=r_j(m-1)$. We now show that $\mu$ is delivered in $r$ in round $m$ iff it is delivered in ${r'}$ in round $\shift{m,t_j}$. 
  \begin{itemize}
      \item If $\mu$ is delivered in round $m$ of $r$ then by Claim 2 we have that $|\mu,\shift{m_h,t_h}|\in chan_{hj}$ at time $\shift{m,t_j}-1$ in ${r'}$ so $\mu$ is delivered in round $\shift{m}$ of ${r'}$ as well. 
      \item Otherwise, i.e., $\mu$ is not delivered in round $m$ of $r$. Assume by way of contradiction that~$\mu$ is  delivered in round $\shift{m,t_j}$ of ${r'}$.
    So $|\mu,\shift{m_h,t_h}|\in chan_{hj}$ at time $\shift{m,t_j}-1$ in ${r'}$ and thus $\mu$ is sent in round $\shift{m_h,t_h}< \shift{m,t_j}$ of ${r'}$. By the inductive hypothesis, $\mu$ is sent in round $m_h$ of~$r$ . Since $\mu$ is not delivered in round~$m$ of $r$, while $\eta_j(r,m)=\del_j(|\mu,m_h|)$, we have that $\mu$ is delivered in some round $m'<m$ of $r$. So by  Claim 2, $\mu$ must be delivered at time $\shift{m',t_j}<\shift{m,t_j}$ in ${r'}$. Hence, $|\mu,\shift{m_h,t_h}|\notin chan_{hj}$ at time $\shift{m,t_j}-1$ in ${r'}$, contradicting the fact that $\mu$ is delivered in round $\shift{m,t_j}$ of ${r'}$.
  \end{itemize}
    We thus obtain that $r_j(m)={r'}_j(\shift{m,t_j})$, and that the same actions (none in this case) and the same messages are delivered in round $m$ of $r$ and in round $\shift{m,t_j}$ of~${r'}$.
		\end{itemize}
\end{proof}

\section{Operations}\label{sec:operations}
To capitalize on the power of \Cref{thm:paslemma2.1}, we now set out to show how operations on distributed objects can be rearranged while maintaining local equivalence. We consider operations that are associated with individual processes. An operation $\Op$ of type~$O$%
\footnote{While processes are typically able to perform particular types of operations on concurrent objects, such as reads, writes, etc., many different instances of an operation may appear in a given run.  Every instance of an operation has a type.}
starts with an invocation input $\inv_i(O,\arg)$ from the environment to process~$i$, and ends when process~$i$ performs a matching response action   $\ret_i(O,\arg)\in Act_i$. 
Operation invocations in our model are nondeterministic and asynchronous --- the environment can issue them at arbitrary times.\footnote{We assume for simplicity that following an $\inv_i$, the environment will not issue another $\inv_i$ to the same process before~$i$ has provided a matching response.} 
Operations can have invocation or return parameters, which appeared as $\arg$ in the above notation. E.g., a write invocation to a register will have a parameter $v$ (the value to be written), while the response to a read on the register will provide the value~$v'$ being read. 


We say that an operation $\X$ occurs between nodes $\theta=\node{i,t}$ and $\theta'=\node{i,t'}$ in $r$ if $\X$'s invocation by the environment (of the form $\inv_i(\cdot)$) occurs in round~$t$ in $r$ and process $i$ performs $\X$'s response action in round~$t'$. In this case we denote $\X.s\triangleq\theta$ and 
$\X.e\triangleq\theta'$, and use $t_{\X.s}(r)$ to denote the operation's starting time~$t$ and $t_{\X.e}(r)$ to denote its ending time~$t'$. When the run is clear from the context we do not precise it.
An operation $\Op$ is {\em completed} in a run~$r$ if $r$ contains both the invocation and response of~$\Op$, otherwise $\Op$ is {\em pending}. Observe that in a crash prone environment, it is not possible to guarantee that every  operation completes, since once a process crashes, it is not able to  issue a response. 
\begin{definition}[Real-time order and concurrency]
   For two operations $\X$ and $\Y$  in $r$ we say that $\X$ {\em precedes} $\Y$ in~$r$, denoted $\X<_{r}\Y$,  if $t_{\X.e}(r)<t_{\Y.s}(r)$, i.e., if~$\X$ completes before~$\Y$ is invoked. 
        If neither $\X$  precedes $\Y$ nor $\Y$ precedes $\X$, then $\X$ and $\Y$ are considered {\em concurrent} in~$r$. 
Finally, $\X$ is said to {\em run in isolation} in~$r$ if no operation is concurrent to $\X$ in $r$.
   
\end{definition}

\begin{definition}[Message chains among operations]
     We write $\X\opmc_r\Y$ and say that there is a message chain between the operations $\X$ and $\Y$ in $r$ if $\X.s\mc_r\Y.e$. 
\end{definition}
Notice that $\X\opmc_r\Y$ does not imply that $\X$ happens before~$\Y$ in real time (i.e., it does not imply that $\X<_r\Y$). Rather, it only implies  that  $\Y$ does  not end before $\X$ starts (i.e., $\Y\not<_r\X$). Moreover, 
while 
`$\mc_r$' among individual nodes is  transitive, 
 `$\opmc_r$' among operations is not. 

An operation $\X$ of~$i$ in the run~$r$ 
is said to \emph{correspond} to operation~$\X'$ of $j$ in $r'$,  denoted by $\X\sim\X'$, if $i=j$ (they are performed by the same process), $\X.s\sim\X'.s$ and $\X.e\sim \X'.e$. 
Note that for locally equivalent runs $r\loc r'$, for every operation~$\X$ in~$r$ there is a corresponding operation $\X'$ in~$r'$. In the sequel, we will often refer to corresponding operations in different runs by the same name.
Observe that, by the definition of $\opmc_r$ and \Cref{lem: loc_eq}, if $\X\opmc_r\Y$ and $r\loc r'$ then $\X\opmc_{r'}\Y$. 

We are now ready to use \Cref{thm:paslemma2.1} to show that if a run does not contain a message chain from one operation to another operation, then  operations in the run can be reordered so that the former operation takes place strictly after the latter one. More formally:  

\begin{theorem}[Moving one operation ahead of the other]\label{thm:reordering}
Let $\X$ and~$\Y$ be two operations in a run $r$.
If $\Y$ completes in~$r$ and $\X\not\opmc_r \Y$,  then there exists a run $r'\loc r$ in which both  (i) $\Y<_{r'}\X$ and  (ii) 
$\X<_{r'}\Z$ holds for every completing operation~$\Z$ in~$r$ such that $\X<_r\Z$ and $\Z\not\opmc_r\Y$.
\end{theorem}

\begin{proof}
    Let $r'$ be the run built in the proof of \Cref{thm:paslemma2.1} wrt.\ the run $r$ with pivot $\theta=\Y.e$ and delay $\Delta=t_{\Y.e}(r)-t_{\X.s}(r)+1$. By \Cref{thm:paslemma2.1} we have that $r\approx r'$, so each process performs the same operations and in the same local order. By the assumption, $\X\not\opmc_r \Y$, i.e., $\X.s\not\mc_r \Y.e$, so $\X$ is moved forward by $\Delta$ while $\Y$ happens at the same real time in both $r$ and~$r'$. We thus have that $\Y<_{r'}\X$ because 
    \[t_{\X.s}(r')=t_{\X.s}(r)+\Delta=t_{\X.s}(r)+t_{\Y.e}(r)-t_{\X.s}(r)+1=t_{\Y.e}(r)+1=t_{\Y.e}(r')+1>t_{\Y.e}(r').\]
    Finally, let $\Z$ be an operation in~$r$ such that $\Z\not\opmc_r\Y$ and $\X<_r\Z$. 
    Since $\Z\not\opmc_r\Y$, the real times of both $\X.e$ and $\Z.s$ in $r'$ are shifted by $\Delta$ relative to their times in~$r$.  Thus, $\X<_r\Z$ implies that~$\X$ ends before $\Z$ starts in $r'$ also, i.e., $\X<_{r'}\Z$.
\end{proof}
\section{Registers and Linearizability}\label{sec:registers}
    A  {\em register} is a shared object that supports two types of operations: {\em reads} $R$ and {\em writes} $W$.  We focus on implementing a MWMR (multi-writer multi-reader) register, in which every process can perform  reads and writes, in an asynchronous message-passing system. 
Simulating a register in an asynchronous system has a long tradition in distributed computing, starting with the work of \cite{ABD}. When implementing registers in the message passing model, one typically aims to mimic the behaviour of an atomic register. A register is called {\em atomic} if its read and write operations are instantaneous, and each read operation returns the value written by the most recent write operation (or some default initial value if no such write exists). The standard correctness property required of such a simulation is Herlihy and Wing's {\em linearizability} condition \cite{HerlihyLineari}. 
   Roughly speaking, an object implementation is linearizable if, although operations can be concurrent, operations behave as if they occur in a sequential order that is consistent with the real-time order in which operations actually occur: if an operation $\Op$ terminates before an operation $\Op'$ starts, then $\Op$ is ordered before $\Op'$.
   More formally:
   
   We denote by $\inv_i(W,v)$ the invocation of a write operation of value $v$ at process~$i$ and by $\ret_i(W)$ the response to a write operation. (Recall that the invocation is an external input  that process~$i$ receives from the environment, while the response is an action that~$i$ performs.)  Similarly, $\inv_i(R)$ denotes the invocation of a read operation at process~$i$ and by $\ret_i(R,v)$ the response to a read operation returning value $v$. We say that an invocation $\inv_i(\cdot)$ and a response $\ret_i(\cdot)$ are {\em matching} if they both are by the same process and in addition, they both are invocation and response of an operation of the same type.
   \begin{definition}[Sequential History]
       A sequential history is a sequence $H= S_0,S_1,...$ of  invocations and responses in which the even numbered elements $S_{2k}$ are invocations and the odd numbered ones are responses, and where $S_{2k}$ and $S_{2k+1}$ are matching invocations and responses whenever $S_{2k+1}$ is an element of~$H$.  
   \end{definition}
We use the following notation:
\begin{notation}
    Let $H$ be a sequential history and let $\X,\Y$ be two operations in $H$. We denote $\X<_{H}\Y$ the fact that $\X$'s response appears before $\Y$'s invocation in $H$.
\end{notation}
    \begin{definition}
        An {\em atomic register history} is a sequential history~$H$ in which every  read operation returns the most recently written value, and if no value is written before the read, then it returns the default value~$\bot$. 
    \end{definition}

\begin{definition}[Linearization]
       A {\em  linearization of a run $r$} is an atomic register history $H$ satisfying the following. 
       \begin{itemize}
           \item The elements of~$H$ consist of the invocations and responses of all completed operations in~$r$, possibly some invocations of pending operations in~$r$, and for each invocation of a pending operation that appears in~$H$, a matching response.  
         \item If $\X<_r\Y$ and the invocation of $\Y$ appears in $H$, then $\X<_H\Y$.
       \end{itemize}
\end{definition}
\begin{definition}[Linearizable Protocols]
     $P$ is a (live) linearizable atomic register protocol (\,\larp) if 
     for every run~$r$ of~$P$:
     \begin{itemize}
         \item every operation invoked at a nonfaulty process in~$r$ completes, and 
     \item there exists a linearization of~$r$ as defined above. 
     \end{itemize} 
    \end{definition}

    Unless explicitly mentioned otherwise, all of the runs~$r$ in our formal statements below are assumed to be runs of an \larp~$P$.


\section{Communication Requirements for Linearizable Registers}\label{sec:XaYb}
In this section, we study the properties of linearizable atomic register protocols in the asynchronous message passing model.
Since linearizability is local \cite{HerlihyLineari}, it suffices to focus on implementing a single register, since a correct implementation will be compatible with linearizable implementations of other registers and objects. 
We assume for ease of exposition that a given value can be written to the register at most once in any given run. (It follows that if the value~$v$ is written in~$r$, we can denote the write operation by $\W(v)$). 

 We say that an operation $\X$ is a \emph{$v$-operation} 
 and write $\X v$ if (i) $\X$ is a read that returns value $v$, or (ii) $\X$ is a write operation writing~$v$. 
 In every linearization history of a run~$r$ of an \larp, a read operation 
returning a value $v\ne\bot$ must be preceded be an operation writing the value~$v$. A direct application of \Cref{thm:reordering} allows us to formally prove that, as expected,  a read operation returning~$v$ must receive a message chain from the operation writing~$v$: 
\begin{lemma}\label{lem: W_to_R}
    If a read operation $\X v$ in~$r$  returns a value $v\ne\bot$ then $\W(v) \opmc_r\X v$.
\end{lemma} 
\begin{proof}
     
    Let~$r$ be a run of a \larp~$P$, and assume by way of contradiction that there is an operation $\X v$ with $v\ne\bot$ in~$r$  such that $\W(v)\not\opmc_r\X v$. Since $\X v$ is assumed to return~$v$, it completes in~$r$. 
    Applying \Cref{thm:reordering} wrt.\ $\X =\X v$ and $\Y=\W(v)$ we obtain a run $r'\loc r$ such that $\X v<_{r'}\W(v)$. By \Cref{lem: loc_eq}(ii) we have $r'$ is a run of~$P$ as well. It follows that~$r'$ must have a linearization~$H$. But by linearizability, $H$ must be such that $\X v<_H \W(v)$. 
    Since~$v$ is written only once, there is no write of~$v$ before $\X v$ in~$H$, contradicting the required properties of a linearization. 
\end{proof}

\Cref{lem: W_to_R} proves an obvious connection: For a value to be read, someone must write this value, and the reader must receive information that this has occurred. 
But as we shall see, linearizability also forces the existence of other message chains; indeed,  most pairs of operations in an execution must be related by a message chain.

A straightforward standard but very useful implication of linearizability for atomic registers is captured by the following lemma. 
\begin{lemma}[no $a$-$b$-$a$]\label{lem:abc}
     Let 
     $\Xa<_r\Yb<_r\Z c$ be three completing operations in a run $r$ of a \larp~${P}$. If $a\neq b$ then $a\neq c$.
\end{lemma}
\begin{proof}

    We first show the following claim:
\begin{claim}\label{obs:linpoint}
    Let $\R v$ be a completing read operation occurring in $r$ and let $H$ be a linearization of $r$. Then (i) $\W (v)<_H\R v$, and moreover (ii)~there is no value~$v'\neq v$ s.t. $\W (v)<_H\W(v')<_H \R v$.
\end{claim}
\begin{proof}
    Recall that the sequential specification of a register states that a read must return the most recent written value. The fact that the value $v$ must have been written implies (i). The fact that it is the last written value linearized before $\R v$ implies (ii).
\end{proof}
    Returning to the proof of \Cref{lem:abc}, let $H$ be a linearization of $r$. Clearly, the real time order requirement of linearizability implies that $\Xa<_H\Yb<_H\Z c$. By \Cref{obs:linpoint} (i), we have that $\W(a)\leq_H\Xa$ and $\W(b)\leq_H\Yb$. Combining these inequalities with \Cref{obs:linpoint} (ii), we obtain that $\W(a)\leq_H\Xa<_H\W(b)\leq_H\Yb\leq_H\Z c$. 
    If $\Z c$ is a write operation then $a\neq c$ results from the fact that $\W(a)<_H \Z c$ and that the value $a$ can be written at most once in $r$. If $\Z c$ is a read operation, then it cannot return $a$ since the value $a$ is not the last written value before~$\Z c$ (since $\W(a)<_H\W(b)$).
   
\end{proof}
Lemmas \ref{lem: W_to_R} and \ref{lem:abc} explain the second communication round of the ABD algorithm \cite{ABD}, also known as {\em Write-Back}: Roughly speaking, the Write-Back of a read $\R$ returning value $v$ guarantees that the reader {\em knows} that for every future read $\R'$, the run will contain a message chain from $\W(v)$ through $\R$ to $\R'$.

Based on \Cref{thm:reordering} and \Cref{lem:abc}, we are now in a position to prove our most powerful result about  linearizable implementations of atomic registers, which shows that they must create message chains between operations of all types: Reads to writes, writes to writes, reads to reads and writes to reads.  Intuitively, \Cref{thm:winningread} shows that if a value~$b$ is read, then every $b$-operation must be reached by a message chain from all other earlier operations.
\begin{theorem}[Linearizability entails message chains]\label{thm:winningread}

   Let $\R b$ be a completing read operation in~$r$ and let $\Yb$ be a $b$-operation that completes in~$r$ such that  $\R b \not\opmc\Yb$. 
   Then for every $c\ne b$ and operation $\X c<_r\R b$, the run~$r$ contains a message chain $\X c\opmc_r\Yb$.
\end{theorem}
\begin{proof}
    Assume by way of contradiction that there is an operation $\X c<_r\R b$ such that $\X c\not\opmc_r\Yb$. First notice that all three operations $\X c$, $\Yb$ and $\R b$ complete in $r'$, since $\R b$ and~$\Yb$ complete by assumption and $\X c<_r\R b$.  We apply \Cref{thm:reordering} wrt.\ $\X=\X c$ and $\Y=\Yb$ and obtain a run $r'\loc r$ such that $\Yb<_{r'}\X c$. Moreover, since $\X c<_r\R b$, we also have by \Cref{thm:reordering} (ii) that $\X c<_{r'}\R b$. We thus obtain $\Yb<_{r'}\X c<_{r'}\R b$ for values $b\neq c$.  This contradicts \Cref{lem:abc}, completing the proof.
\end{proof}

Intuitively, \Cref{thm:winningread} shows that read or write operations involving a value that is actually read (i.e., returned by a read operation) must receive message chains from practically all earlier operations. We can show that the same can be true more broadly, e.g., even for a completing write operation $\W(v)$ where~$v$ is never read in the run. 

\begin{corollary}
\label{thm:isolation}
     Let $\Xa<_r\Yb$ and assume that~$\Yb$ completes in~$r$. If $\Yb$ runs in isolation in~$r$ and $a\neq b$, then $\Xa\opmc_r \Yb$.
\end{corollary}
\begin{proof}
Let $r$ be a run satisfying the assumptions. There exists a run~$r'$ such that (i) $r'$ is identical to~$r$ up to $t_{\Yb.e}(r)$ (in particular, $r'(m)= r(m)$ for all $0\leq m \leq t_{\Yb.e}$), and (ii) there is an invocation of a read operation $\R$ in round $t_{\Yb.e}+1$ of~$r'$, at a process~$i$ that is nonfaulty in $r'$. Since~$i$ is nonfaulty, $\R$ completes in~$r'$. 
    Moreover, since~$\Yb$ runs in isolation and  $\R$ starts after $\Yb$ ends, the value returned by~$\R$ must be $b$. We obtain a run $r'$ in which $\Yb<_{r'}\R b$ and $\Xa<_{r'}\R b$ with $a\neq b$. So by \Cref{thm:winningread} we have that $\Xa\opmc_{r'} \Yb$. Since $r'(m)= r(m)$ for all $0\leq m \leq t_{\Yb.e}$ it follows that $\Xa\opmc_{r}\Yb$, as claimed.
\end{proof}

\section{Failures and Quorums}\label{sec:quorumsandfailures}

By assumption, invocations of reads and writes to a register are  spontaneous events, which is modeled by assuming that they are determined by the adversary (or the environment in our terminology) in a  nondeterministic fashion. Intuitively, in  a completing register implementation, the adversary can at any point wait for all operations to return and then perform a read. Suppose that this read operation is invoked at time~$t$ and that the value it returns is~$v$. Then, by \Cref{thm:winningread}, the resulting run~$r$ must contain  message chains $\X\opmc_r \W(v)$ from \emph{every} operation~$\X$ that completed before time~$t$ to the write operation $\W(v)$. 
Therefore, before it can complete, every operation~$\X$ must ensure that message chains from~$\X$ to future operations can be constructed. 
There are several ways to ensure this in a reliable system. 
One way is by requiring the process on which~$\X$ is invoked to construct a message chain to all other processes before $\X$ returns. This essentially requires a broadcast to all processes that starts after~$\X$ is invoked.   
Another way to ensure this is by having every transaction~$\Y$  coordinate a convergecast to it from all processes, that is initiated after~$\Y$ is invoked. Each of these can be rather costly. A third, and possibly more cost effective way can be to assign a distinguished coordinator process~$c$ for the register object, and ensure that every operation~$\X$ creates a message chain to~$c$ that is followed by a message chain back from~$c$ to the process invoking~$\X$. 
Notice that none of these strategies can be used in a system in which one or more processes can crash: After a crash, neither the broadcast nor the convergecast would be able to complete. Similarly, a coordinator~$c$ as described above would be a single point of failure, and once it crashes no operation could complete.  


We now show that in a system in which up to~$f$ processes can crash, \Cref{thm:winningread} implies that an operation must complete round-trip communications with at least~$f$ other processes before it can terminate. We proceed as follows.


\begin{definition}
We say that a process~$p$ \emph{observes} a completed operation~$\X$ in a run~$r$ if~$r$ contains a message chain from $\X.s$ to $\node{p,t_{\X.e}}$. (The message chain reaches~$p$ by the time operation~$\X$ completes.)
Process~$p$ is called a \emph{witness for~$\X$} in~$r$ if $r$ contains a message chain from~$\X.s$ to $\X.e$ that contains a $p$-node $\theta=\node{p,t}$. 
\end{definition}

\begin{lemma}\label{lem: observers}
     Let~$P$ be an $f$-resilient \larp, 
     and let~$\X$ be a completed  operation in a run~$r$  of~$P$. Then more than~$f$ processes must observe~$\X$ in~$r$.
\end{lemma}
\begin{proof}
    Assume, by way of contradiction, that no more than~$f$ processes observe~$\X$ in~$r$. 
    Let~$r'$ be a run of~$P$ that coincides with~$r$ up to time~$t_{\X.e}$, in which all processes that have observed~$\X$ fail from round $t_{\X.e}+1$ (and no other process crashes), in which all operations that are concurrent with~$\X$ complete and, after they do, a write operation $\W(v)$ (for a value~$v$ not previously written)  runs in isolation, followed by a completed read. 
Since all processes that observed~$\X$ in~$r'$ crash before $\W(v)$ is invoked, $\X\not\opmc_{r'}\W(v)$. 
 The read returns~$v$, and so \Cref{thm:winningread} implies that $\X\opmc_{r'}\W(v)$, contradiction. 
\end{proof}

We can now show that in $f$-resilient \larp's, every operation must perform at least one round-trip communication to all members of a quorum set of size at least~$f$. Formally: 

\begin{theorem}\label{thm: witnesses}
         Let~$P$ be an $f$-resilient \larp, and let~$\X$ be a completed  operation in a run~$r$  of~$P$. Then~$r$ must contain more than~$f$ witnesses for~$\X$. 
\end{theorem}
\begin{proof}
 Assume by way of contradiction that there is a run~$r$ of~$P$ that contains $\le f$ witnesses for~$\X$. 
 Notice that for every witness~$p$ for~$\X$ in~$r$ there must be a node $\node{p,t}\mc_r\X.e$. 
 We apply \Cref{thm:paslemma2.1} to~$r$ with pivot $\X.e$ and delay $\Delta=t_{\X.e}-t_{\X.s}+1$, to obtain a run $r'\loc r$. By \Cref{lem: loc_eq}(iii) the run~$r'$ is a run of~$P$.  
 By choice of $\Delta$, only processes with nodes in $\past_{r'}(\X.e)$ can observe $\X$ in $r'$, so every observer of~$\X$ must be a witness for~$\X$. 
 By construction $\past_r(\X.e)=\past_{r'}(\X.e)$, and so there are no more than~$f$ witnesses for~$\X$ in~$r'$. 
 It follows that no more than~$f$ processes observe~$\X$ in~$r'$, contradicting \Cref{lem: observers}
.\end{proof}
The ABD algorithm requires the number of processes to satisfy $n\geq 2f+1$ \cite{ABD}. This ensures that  every two sets of $n-f$ processes intersect in at least one process, i.e., each operation communicates with a quorum set. We remark that although \Cref{thm: witnesses} implies the need to communicate with quorum sets, the Write-Back round is not {\em always} necessary. If a  reader of~$v$ receives message chains from all processes that are in a quorum set that $\W(v)$ communicated with in the first round, then the message chains of \Cref{lem: W_to_R} can be  guaranteed without the Write-Back. The algorithm of \cite{fastAtomicRegister04} is based on this type of observation. In addition, strengthening the results of \cite{naserpastoriza2023}, our work implies that the  channels that are shown to exist in \cite{naserpastoriza2023} must in fact be used to interact with quorums.

\bibliography{ref.bib,ref3.bib, z2_DISC23.bib, z1.bib}
\end{document}